\documentclass[preprint,5p]{elsarticle}

\newcommand{\myTitle}{Membership(s) and compliance(s) with class-based graphs}

\newcounter{defNumber}

\newenvironment{definition}[1]
{
\medskip
\refstepcounter{defNumber}\noindent\textbf{Definition~\arabic{defNumber} (#1)}\hspace{1ex}}
{
\smallskip
}

\newenvironment{proof}[1][Proof]{\textbf{#1.} }{\ \rule{0.5em}{0.5em}}

\newcommand{\backbullet}{\ensuremath{\,\,^{\bullet}\!\!}}
\newcommand{\topbullet}{\ensuremath{\,\,\,^{^{\bullet{}}}\!\!\!\!\!\!\!\!}}
\newcommand{\backcirc}{\ensuremath{^{\,\,\circ}\!\!}}
\newcommand{\topcirc}{\ensuremath{\,\,\,^{^{\circ{}}}\!\!\!\!\!\!\!\!}}

\newcommand{\NSM}{\ensuremath{\topcirc\subset}}
\newcommand{\NFM}{\ensuremath{\topbullet\subset}}

\newcommand{\ASM}{\ensuremath{\backcirc\subset^{\circ}}}
\newcommand{\ALM}{\ensuremath{\backbullet\subset^{\circ}}}
\newcommand{\ARM}{\ensuremath{\backcirc\subset^{\bullet}}}
\newcommand{\AFM}{\ensuremath{\backbullet\subset^{\bullet}}}

\newcommand{\CR}{\ensuremath{\leftmodels}\xspace}
\newcommand{\PCR}{\ensuremath{\dashv}\xspace}
\newcommand{\FCR}{\ensuremath{\leftModels}\xspace}

\newcommand{\ie}{i.e., }
\newcommand{\eg}{e.g., }

\newcommand{\Romeo}{\texttt{Romeo}\xspace}
\newcommand{\Juliet}{\texttt{Juliet}\xspace}
\newcommand{\Tybalt}{\texttt{Tybalt}\xspace}
\newcommand{\Montague}{\texttt{Montague}\xspace}

\newcommand{\Mercutio}{\texttt{Mercutio}\xspace}

\newcommand{\male}{\texttt{male}\xspace}

\newcommand{\arcKilled}{\texttt{has} \texttt{killed}\xspace}
\newcommand{\arcCousin}{\texttt{cousin}\xspace}
\newcommand{\arcSuicide}{\texttt{commit} \texttt{suicide}\xspace}
\newcommand{\arcSuicideToo}{\texttt{commit} \texttt{suicide} \texttt{too}\xspace}
\newcommand{\arcFeelings}{\texttt{feelings}\xspace}

\newcommand{\ClassMrMontague}{\texttt{Mr.} \texttt{Mon\-ta\-gue}$^\alpha$\xspace}
\newcommand{\ClassCapulet}{\texttt{Ca\-pu\-let}$^\alpha$\xspace}
\newcommand{\ClassMissCapulet}{\texttt{Miss} \texttt{Capulet}$^\alpha$\xspace}

\newcommand{\arcClassSuicide}{\texttt{commit} \texttt{suicide}$^\alpha$\xspace}
\newcommand{\arcClassKilled}{\texttt{has} \texttt{killed}$^\alpha$\xspace}
\newcommand{\arcClassCousin}{\texttt{cousin}$^\alpha$\xspace}
\newcommand{\arcClassFeelings}{\texttt{feelings}$^\alpha$\xspace}

\newcommand{\propHouse}{\texttt{house}\xspace}
\newcommand{\propKilling}{\texttt{killing}\xspace}
\newcommand{\propSex}{\texttt{sex}\xspace}

\usepackage{tikz} %%Generate vector graphics from within LaTeX
\usepackage{pgf}
\usetikzlibrary{arrows,automata}
\usetikzlibrary{positioning}
\tikzset{
    object/.style={
           rectangle,
           rounded corners=4mm,
           fill=black!10,
           draw=black, very thick,
           minimum height=2em,
           inner sep=1pt,
           text centered,
           font=\small,
           },
    arc/.style={
           rectangle,
           rounded corners=1mm,
           draw=black!50, thick,
           fill=white,
           minimum height=2em,
           inner sep=1pt,
           text centered,
           font=\small,
           },
}

\usepackage{amsmath}%
\usepackage{amsfonts}%
\usepackage{fourier}
\usepackage{MnSymbol}
\usepackage{graphicx}
\usepackage[T1]{fontenc}

\usepackage{xspace}

\usepackage{url}

\usepackage{epstopdf}
\usepackage[pdftex,
	pdfauthor={Willy Picard},
	pdftitle={\myTitle},
  colorlinks,
	urlcolor=black,
	citecolor=black,
	linkcolor=black
	]{hyperref}

\journal{Information Processing Letters}

\begin{document}
\begin{frontmatter} 
\title{\myTitle{}}

\author[UEP]{Willy Picard}
\ead{picard@kti.ue.poznan.pl}
\address[UEP]{Department of Information Technology,
Pozna\'{n} University of Economics, 
al. Niepodleglo\'{s}ci 10, 61-875 Pozna\'{n}, 
Poland}

\begin{abstract}
Besides the need for a better understanding of networks, there is a need for prescriptive models and tools to specify requirements concerning networks and their associated graph representations. We propose class-based graphs as a means to specify requirements concerning object-based graphs. Various variants of membership are proposed as special relations 
between class-based and object-based graphs at the local level, while various variants of compliance are proposed at the global level.
\end{abstract}


\end{frontmatter}

\section{Introduction}

During the last decade, the Web has been transformed from a siloed information medium into a highly dynamic network of information, created by individuals and organizations mostly in a participatory manner. This network is currently referred to as the Web 2.0~\cite{blog:oreilly:2005}.
 
A reason for such a transformation may be the better understanding of the structure and the functioning of networks, especially small-world networks~\cite{article:duncan:1998,article:barabasi:1998}.
%,book:buchanan:2003
It has been shown that many networks, from Web pages~\cite{article:albert:1999:web} to food webs~\cite{article:montoya:2002}, from research paper co-authoring~\cite{article:newman:2001} to human brain functional networks~\cite{article:eguiluz:2005}, share a set of common characteristics, \eg their average shortest path is relatively low, while their clustering coefficient is rather higher than in random networks.

Networks have also been studied as regards their dynamics. As an example, network percolation, with its potential application to explain disease epidemics and gossip propagation, has received significant attention~\cite{article:callaway:2000, conf:doerr:2011}.
%article:satorras:2001, 

%~\cite{www:facebook} ~\cite{www:twitter}
Social websites, such as Facebook or Twitter, are playing a major role on the Web 2.0. These sites support social networks, \ie networks of individuals and organizations linked by their relationships. Once again, the characteristics of social websites and the social networks they support is the subject of many research works~\cite{article:lewis:2008,article:richter:2011}. 
%article:iribarren:2011,

Besides the descriptive approach of the works mentioned above, prescriptive tools are needed. Not only should networks be understood, but tools to specify constraints on networks are required. With the ubiquity of networks, tools are needed to check if a chosen subset of a given network satisfies a predefined set of constraints. These constraints should concern both the nodes of the network and the arcs between them. To draw an analogy with collective sports, it is important not only to understand how a team is performing, but also to be able to define requirements about the various players and their potential relations. The coach would then be able to check if a given set of players satisfies his/her expectations defined as requirements.

The problem addressed in this paper may be stated as follows: how to specify a type of networks with constraints on both nodes and arcs, and how to define the concept of compliance of a given network with these constraints. Many applications of this problem may be found, such as the establishment of the cast of a movie, the specification of emergency crews, the definition of a set of chemical substances needed for a given chemical reaction, the specification of the set of web services required to implement a given service-oriented application, and the definition of crews in hospitals for surgical operations. 

For the sake of readability and conciseness, an simplified example based on William Shakespeare's tragedy \emph{Romeo and Juliet} is presented in this paper. Besides the criterion of succinctness, the choice of this example is guided by the assumption that the popularity of this play will ease the understanding of the various illustrative networks presented in the rest of the paper, networks consisting of characters from \emph{Romeo and Juliet}. 

Addressing this problem encompasses two main issues. First, the development of tools supporting the definition of constraints on networks implies both a representation for networks and a representation for the constraints. Second, different types of relations between the networks and the constraints may occur and should precisely be defined.

Three IT areas partially address the proposed problem: object-oriented languages, database schemata, and ontologies. Object-oriented languages rely on the concept of class to mo\-del constraints on objects~\cite{book:pierce:2002}. A network of classes may therefore constrain a network of objects. Similarly, database sche\-mata, either in relational, object-oriented, or XML da\-ta\-bases, constraint the da\-ta\-base~\cite{book:dbReadings:2005}. Finally, classes in ontologies constrain individuals~\cite{book:gargouri:2010}. However, in these three approaches, classes and database schema have to exist to be instantiated as objects, data, and individuals. Therefore, classes and database schema have to precede objects. In the case of social networks (or former mentioned collective sports), the network usually exists before the constraints do. Additionally, in these three approaches, a limited support for arcs is proposed. In object-oriented languages, relations between classes are limited to \texttt{has-a} and \texttt{is-a} relations, via class attributes and inheritance. In relational database, the only mechanism to connect relations are joins. Finally, similarly to  the object-oriented approach, the relations between classes in ontologies are limited to class properties and inheritance. It should be possible to specify constraints on arcs in a more subtle manner, encompassing a more complex representation of the arcs among nodes of a network.

In this paper, we propose to represent networks as object-based graphs. Constraints on networks may then be represented as class-based graphs. We further define the concepts of membership and compliance. Membership concerns single objects and classes, while compliance concerns the whole graphs. Various variants of membership and compliance are proposed in this paper. Related works are discussed in Section~\ref{sec:relatedWorks}.

Our main contribution reported here are: a formal definition of object-based and class-based graphs; the identification and formal definition of various types of membership for nodes and arcs; the identification and formal definition of various types of compliance of object-based graphs with class-based graphs.

\section{Object-based graphs}
The concepts of object-based graphs and related class-based graphs are based on the concepts of \emph{object} and \emph{class}. For the sake of precision, clear definitions are mandatory in light of the different meanings of these terms in various research communities.

An \emph{object} is a set of properties $o=\{p\}$. A \emph{property} $p$ is a pair $\langle n,v_n \rangle$, where $n$ is the name of the property and $v_n$ is the value of the property. The value of a property may be a literal or an object.

An \emph{object-based node}, denoted $n$, is an object that does not contain properties named neither \texttt{src} nor \texttt{dst}. 

Note that the property name \texttt{src} (resp.~\texttt{dst}) is reserved to the source (resp.~destination) of arcs.

An \emph{object-based arc}, denoted $a$, is an object  that contains at least two properties named \texttt{src} and \texttt{dst} whose values are object-based nodes, \ie 
$\exists \big(\langle \texttt{src}, n_{\texttt{src}} \rangle, 
\langle \texttt{dst}, n_{\texttt{dst}} \rangle \big) \in a \times a$, with $n_{\texttt{src}}$ and $n_{\texttt{dst}}$ being object-based nodes.

\begin{definition}{Object-Based Graph}
An \emph{object-based graph} $g= \langle N, A \rangle$ is a graph whose nodes are object-based nodes, and arcs are object-based arcs, with values of the properties named \texttt{src} and \texttt{dst} being nodes of the graph, \ie
$\forall a \in A$, $\langle \texttt{src}$, $n_{\texttt{src}} \rangle \in a \Rightarrow n_{\texttt{src}} \in N$, and $\forall a \in A$, $\langle \texttt{dst}$, $n_{\texttt{dst}} \rangle \in a \linebreak[0]\Rightarrow n_{\texttt{dst}} \in N$.
\end{definition}

\sloppy
Object-based arcs are further denoted $a = \langle n_{\texttt{src}},$ $n_{\texttt{dst}},P\rangle$, where $n_{\texttt{src}}$ is the value of the property named \texttt{src}, $n_{\texttt{dst}}$ is the value of the property named \texttt{dst}, and $P$ is the set of remaining properties of the arc.
\fussy

\begin{figure}[htp]
	\centering
		%%%%%%%%%%%%%%%%%%%%%%%%%%%
%%
%% Object-based graph
\begin{tikzpicture}[->,>=stealth']

	\node[object] (O1)
	{   \begin{tabular}{l} 	% content
		  \textbf{Romeo}\\
		  \hspace{0.2cm}{$\langle$name, Romeo$\rangle$}\\
		  \hspace{0.2cm}{$\langle$house, Montague$\rangle$}\\
      \hspace{0.2cm}{$\langle$sex, male$\rangle$}\\
		  \end{tabular}
 	};

	\node[object,
	above right of=O1,
  yshift=+2cm,
  xshift=+2cm,
  anchor=center] (O2)
	{
 \begin{tabular}{l} 	% content
  \textbf{Tybalt}\\
  \hspace{0.2cm}{$\langle$name, Tybalt$\rangle$}\\
  \hspace{0.2cm}{$\langle$house, Capulet$\rangle$}\\
  \hspace{0.2cm}{$\langle$sex, male$\rangle$}\\
 \end{tabular}
 };

	\node[object,
	below right of=O2,
  yshift=-2.2cm,
  xshift=2cm,
  anchor=center] (O3)
	{
 \begin{tabular}{l} 	% content
  \textbf{Juliet}\\
  \hspace{0.2cm}{$\langle$name, Juliet$\rangle$}\\
  \hspace{0.2cm}{$\langle$house, Capulet$\rangle$}\\
  \hspace{0.2cm}{$\langle$sex, female$\rangle$}\\
  \hspace{0.2cm}{$\langle$age, 13$\rangle$}
 \end{tabular}
 };

 \path 
 		(O1) 	edge[bend left=20]  node[arc,above left=-5mm]{ \begin{tabular}{l}
																														\textit{has killed}\\
																														$\langle$duel, refused$\rangle$\\
																														$\langle$killing, sword$\rangle$
																														\end{tabular}}(O2)
 		(O2) 	edge[bend left=20]  node[arc,above right=-5mm]{\begin{tabular}{l}
																														\textit{cousin}\\
																														$\langle$sibling, cousin$\rangle$\\
																														\end{tabular}}(O3)
		(O1) 	edge[bend left=12]  node[arc,above=1mm]{\begin{tabular}{l}
																														\textit{feelings}\\
																														$\langle$feels,Love$\rangle$\\
																														\end{tabular}}
		(O3) 	
		
		(O3) edge[bend left=12] node[]{} (O1)

		(O1)  edge[loop below,min distance=11mm]    node[arc,below](suic1){\begin{tabular}{l}
																														\textit{commit suicide}\\
																														$\langle$killing, poison$\rangle$
																														\end{tabular}} 
		(O1)
 																														
		(O3)  edge[loop below,min distance=5mm]    node[arc,below]{\begin{tabular}{l}
																														\textit{commit suicide too}\\
																														$\langle$killing, dagger$\rangle$
																														\end{tabular}} 
		(O3)
;
\end{tikzpicture}
	\caption{Example of an object-based graph.}
	\label{fig:objectGraph}
\end{figure}
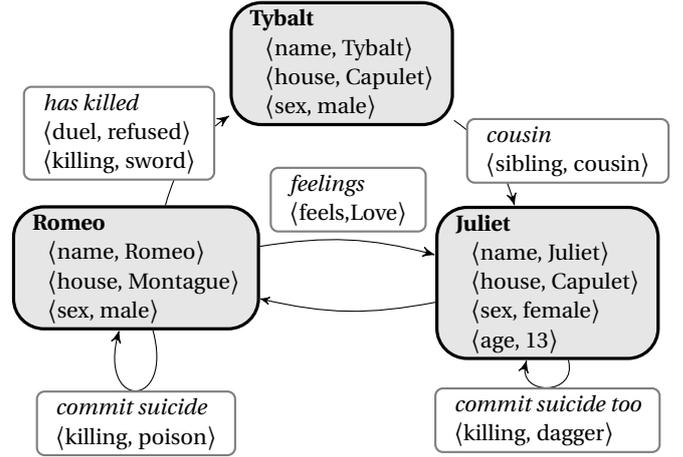

\sloppy
An example of an object-based graph is illustrated in Fig.~\ref{fig:objectGraph}. In this example, the graph consists of three object-based nodes---\Romeo, \Tybalt, and \Juliet---and six object-based arcs---\arcKilled, \arcCousin, \arcSuicide, \arcSuicideToo and twice \arcFeelings. Nodes are represented as greyed rectangles, while arcs are rep\-re\-sent\-ed as arrows and associated white rectangles. Therefore, the node \Romeo consists of three properties $\langle$\texttt{name}, $\Romeo\rangle$, $\langle$\propHouse, $\Montague\rangle$, and $\langle$\propSex, $\male\rangle$. The arc \arcKilled connects the node \Romeo with the node \Tybalt, and consists of two properties $\langle$\texttt{duel}, $\texttt{refused}\rangle$ and $\langle$\propKilling, $\texttt{sword}\rangle$.

\fussy
The arc \arcSuicide is an example of a loop arc connecting the node \Romeo to itself.

It may be noted that a property named \propHouse is an element of the nodes \Romeo, \Tybalt and \Juliet. It has a different value for \Romeo and \Juliet, and a common value for \Tybalt and \Juliet. Similarly, properties named \propKilling are elements of the arcs \arcKilled and \arcSuicide, but with different values.

\section{Class-based graphs}

A \emph{class} is a set of property constraints $c=\{p^\alpha\}$.
A \emph{property constraint} $p^\alpha$ is a pair $\langle n, v^\alpha_n \rangle$, where $n$ is the name of the properties potentially constrained by $p^\alpha$, and $v^\alpha_n$ is a predicate.

A property $p=\langle n,v_n \rangle$ \emph{satisfies} a property constraint $p^\alpha=\langle n', v^\alpha_{n'} \rangle$, denoted $p \succ p^\alpha$, iff $n=n'$ and $v^\alpha_{n'}(v_n)=\texttt{true}$.

Note the `$\alpha$' letter that indicates the class-related character of its associated entities.

\begin{definition}{Class Instance}
An object $o=\bigl\{p=\langle n,v_n \rangle\bigr\}$ is an \emph{instance} of a class $c=\bigl\{p^\alpha=\langle n,v^\alpha_n \rangle\bigr\}$, denoted $o \sqsubset c$,  iff $\forall p^\alpha \in c, \exists p \in o : p \succ p^\alpha$. The $\sqsubset$ predicate is further referred to as \texttt{instanceOf}.
\end{definition}

A \emph{class-based node} $n^\alpha$ is a class that does not contain property constraints named neither \texttt{src} nor \texttt{dst}.

A \emph{class-based arc} $a^\alpha$ is a class containing two property constraints respectively named \texttt{src} and \texttt{dst} whose values are \texttt{in\-stanceOf} predicates associated with source and destination class-based nodes, \ie 
$\exists \big( \langle \texttt{src}$, $\sqsubset n^\alpha_{\texttt{src}}\rangle$, $\langle \texttt{dst}$, $\sqsubset n^\alpha_{\texttt{dst}} \rangle \big)$ $\in a^\alpha \times a^\alpha$, with $n^\alpha_{\texttt{src}}$ and $n^\alpha_{\texttt{dst}}$ being class-based nodes.

\begin{definition}{Class-Based Graph}
A \emph{class-based graph} $g^\alpha=\langle N^\alpha$, $A^\alpha \rangle$ is a graph whose nodes are class-based nodes and arcs are class-based arcs, where values of the property  constraints named \texttt{src} and \texttt{dst} being \texttt{instanceOf} predicates with nodes of the graph, \ie 
$\forall a^\alpha \in A^\alpha$, $\langle \texttt{src}$, $\sqsubset n^\alpha_{\texttt{src}}\rangle \in a^\alpha \Rightarrow n^\alpha_{\texttt{src}} \in N^\alpha$ and
$\forall a^\alpha \in A^\alpha$, $\langle \texttt{dst}$, $\sqsubset n^\alpha_{\texttt{dst}}\rangle \in a^\alpha \Rightarrow n^\alpha_{\texttt{dst}} \in N^\alpha$.
\end{definition}

Arcs in class-based graphs are further denoted $a^\alpha=\langle n^\alpha_{\texttt{src}}$, $n^\alpha_{\texttt{dst}}$, $P^\alpha\rangle$, where the predicate $\sqsubset n^\alpha_{\texttt{src}}$ is the value of the property constraint named \texttt{src}, the predicate $\sqsubset n^\alpha_{\texttt{dst}}$ is the value of the property constraint named \texttt{dst}, and $P^\alpha$ is the set of remaining property constraints of the arc.

\begin{figure}[htp]
	\centering
		%%%%%%%%%%%%%%%%%%%%%%%%%%%
%%
%% Object-based graph
\begin{tikzpicture}[->,>=stealth']

	\node[object] (C1)
	{   \begin{tabular}{l} 	% content
		  \textbf{Mr. Montague$^\alpha$}\\
		  \hspace{0.2cm}{$\langle$house, =Montague$\rangle$}\\
  		\hspace{0.2cm}{$\langle$sex, =male$\rangle$}\\
		  \end{tabular}
 	};

	\node[object,
	above right of=C1,
  yshift=+2cm,
  xshift=+1.9cm,
  anchor=center] (C2)
	{
 \begin{tabular}{l} 	% content
  \textbf{Capulet$^\alpha$}\\
  \hspace{0.2cm}{$\langle$house, =Capulet$\rangle$}
 \end{tabular}
 };

	\node[object,
	below right of=C2,
  yshift=-2cm,
  xshift=+2cm,
  anchor=center] (C3)
	{
 \begin{tabular}{l} 	% content
  \textbf{Miss Capulet$^\alpha$}\\
  \hspace{0.2cm}{$\langle$house, =Capulet$\rangle$}\\
  \hspace{0.2cm}{$\langle$sex, =female$\rangle$}\\
 \end{tabular}
 };

 \path 
 		(C1) 	edge[bend left=20]  node[arc,above left=-6mm]{ \begin{tabular}{l}
																														\textit{has killed$^\alpha$}\\
																											$\langle$killing,=\texttt{true}$\rangle$\\
																														\end{tabular}}
		(C2)
		
 		(C2) 	edge[bend left=20]  node[arc,above right=-6mm]{\begin{tabular}{l}
																														\textit{cousin$^\alpha$}\\
																														$\langle$sibling,=cousin$\rangle$\\
																														\end{tabular}}
		(C3)

		(C1) 	edge[bend left=12]  node[arc,above=1mm]{\begin{tabular}{l}
																														\textit{feelings$^\alpha$}\\
																														$\langle$feels,=Love$\rangle$\\
																														\end{tabular}}
		(C3) 	
		
		(C3) edge[bend left=12] node[]{} (C1)

		(C1)  edge[loop below,min distance=5mm]    node[arc,below]{\begin{tabular}{l}
																														\textit{commit suicide$^\alpha$}\\
																											$\langle$killing,=\texttt{true}$\rangle$\\
																														\end{tabular}} 
		(C1)
 																														;

\end{tikzpicture}
	\caption{Example of a class-based graph.}
	\label{fig:classGraph}
\end{figure}
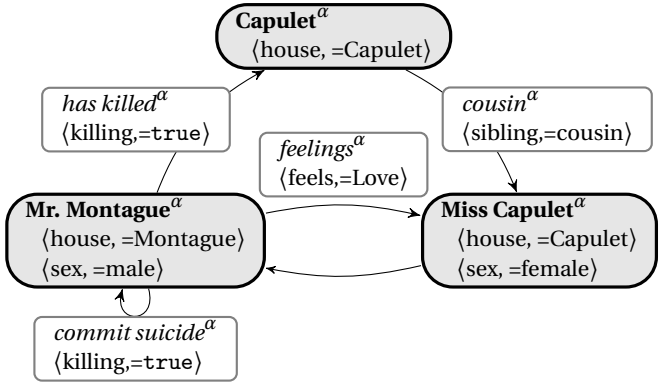

\sloppy
An example of a class-based graph is illustrated in Fig.~\ref{fig:classGraph}. In this example, the graph consists of three class-based nodes---\ClassMrMontague, \ClassCapulet, and \ClassMissCapulet---and five class-based arcs---\arcClassSuicide, \arcClassKilled, \arcClassCousin, and twice \arcClassFeelings. Nodes are represented as greyed rectangles, while arcs are rep\-re\-sent\-ed as arrows and associated white rectangles. Therefore, the node \ClassMrMontague consists of two property constraints $\langle$\propHouse, $=\Montague\rangle$ and $\langle$\propSex, $=\male\rangle$. The arc \arcClassKilled connects the node \ClassMrMontague with the node \ClassCapulet, and consists of one property constraint $\langle$\propKilling, $=\texttt{true}\rangle$.

The object \Romeo is an instance of the class \ClassMrMontague. The two property constraints defined in the class are satisfied by the properties of the object: Romeo is indeed from the Montague house and he is a male. One may notice that the property named \texttt{name}, defined for \Romeo, plays no role when verifying if is an instance of the class \ClassMrMontague. In a similar manner, the property named \texttt{duel} in the \arcKilled object does not play a role in this object being an instance of the class \arcClassKilled.

The object \Juliet is an instance of both the class \ClassCapulet and the class \ClassMissCapulet, while \Tybalt is only an instance of the class \ClassCapulet, and not \ClassMissCapulet (not being a female). One may notice the lack of an \arcClassSuicide related to the class \ClassMissCapulet as deviation from Shakespeare's tragedy. It is a voluntary omission, to illustrate concepts presented in Section~\ref{sec:memberships}.

Finally, one may note that the set of property constraints of the arcs \arcClassKilled and \arcClassSuicide are similar. Even their source nodes, \ie the node associated with the \texttt{src} property constraint, are similar. These two arcs differ only in their destination nodes, \ie the node associated with the \texttt{dst} property constraint.
\fussy

\section{Membership(s)}
\label{sec:memberships}
Although being sufficient to formally define class-based graphs, the concept of class instance does not capture important relationships between class-based and object-based graphs. The concept of membership, in its various variants, is proposed in this section as a means to describe particular local relationships between classes and object in graphs.

\begin{definition}{Node Strict Membership}
A node $n$ is a \emph{strict member} of a class $n^\alpha$, denoted $n \NSM n^\alpha$, iff $n$ is an instance of $n^\alpha$, i.e. $n \sqsubset n^\alpha$.
\end{definition}

Being a class strict member is equivalent with being a class instance. Therefore, \Romeo is a strict member of the \ClassMrMontague class.

Membership of an arc $a = \langle n_{\texttt{src}}, n_{\texttt{dst}}$, $\{p\}\rangle$ to a class may be defined in various ways, depending on the elements of the arcs taken into account. \emph{Arc strict membership}, denoted \ASM, takes into account only the properties $\{p\}$. \emph{Arc left} (resp. \emph{right}) \emph{membership}, denoted \ALM (resp. \ARM), takes into account the properties and the source node (resp. the destination node). Finally, \emph{arc full membership}, denoted \AFM, takes into account the properties and both the source and the destination nodes.

\begin{definition}{Arc Strict Membership}
An arc $a = \langle n_{\texttt{src}},$ $n_{\texttt{dst}}$, $\{p\}\rangle$ is a \emph{strict member} of the class $a^\alpha=\langle n^\alpha_{\texttt{src}}, n^\alpha_{\texttt{dst}},$ $\{p^\alpha\}\rangle$, denoted $a \ASM a^\alpha$, iff $\forall p^\alpha \in a^{\alpha}, \;\exists p\in a$ such that $p \succ p^\alpha$.
\end{definition}

Note that if an arc $a$ is an instance of a class $a^\alpha$, then $a$ is a strict member of $a^\alpha$ as all the properties constraints of $a^\alpha$ are satisfied by the properties of $a$. The opposite is not always true: if an arc $a$ is a strict member of a class $a^\alpha$, then $a$ does not have to be an instance of $a^\alpha$ because the values of the properties \texttt{src} and \texttt{dst} do not have to satisfy the property constraints with the same names. 

Formally,
\begin{equation}
a \sqsubset a^\alpha \quad\Longrightarrow\quad a \ASM a^\alpha.
\label{eq:membershipImpliesInstance}
\end{equation}

\sloppy
The \arcKilled object is an instance of \arcClassKilled. Therefore, \arcKilled is also a strict member of \arcClassKilled. However, although the \arcKilled object is a strict member of the \arcClassSuicide class, it is not an instance of this class: each instance of the \arcClassSuicide class requires two \ClassMrMontague instances as source and destination object.
\fussy

\begin{definition}{Arc Left Membership}
An arc $a=\langle n_{\texttt{src}}$, $n_{\texttt{dst}}$, $\{p\}\rangle$ is a \emph{left member} of the class $a^\alpha=\langle n^\alpha_{\texttt{src}}$, $n^\alpha_{\texttt{dst}}$, $\{p^\alpha\}\rangle$, denoted $a \ALM a^\alpha$, iff 
	$a$ is a strict member of $a^\alpha$ and 
	$n_{\texttt{src}}$ is a strict member of $n^\alpha_{\texttt{src}}$, \ie 
	\begin{equation*}
	a \ALM a^\alpha \quad\Longleftrightarrow\quad a \ASM a^\alpha \;\wedge\; 
	n_{\texttt{src}} \NSM n^\alpha_{\texttt{src}}\;.
	\end{equation*}
\end{definition}

The \arcKilled arc is not only a strict member of \arcClassSuicide, it is also a left member as the source node of \arcKilled, \ie \Romeo, is an instance of the source node of \arcClassSuicide, \ie the \ClassMrMontague class.

\begin{definition}{Arc Class Right Membership}
An arc $a=\langle n_{\texttt{src}}$, $n_{\texttt{dst}}$, $\{p\}\rangle$ is a \emph{right member} of the class $a^\alpha=\langle n^\alpha_{\texttt{src}}$, $n^\alpha_{\texttt{dst}}$, $\{p^\alpha\}\rangle$, 
denoted $a \ARM a^\alpha$, iff 
$a$ is a strict member of~$a^\alpha$ and 
$n_{\texttt{dst}}$ is a strict member of $n^\alpha_{\texttt{dst}}$, \ie
	\begin{equation*}
	a \ARM a^\alpha \quad\Longleftrightarrow\quad a \ASM a^\alpha \;\wedge\; 
	n_{\texttt{dst}} \NSM n^\alpha_{\texttt{dst}}\;.
	\end{equation*}
\end{definition}

\sloppy
The \arcKilled arc is not a right member of \arcClassSuicide. The destination node of the arc \arcKilled, \ie \Tybalt, is not a strict member of the destination node of \arcClassSuicide, \ie the \ClassMrMontague class. 
\fussy

\begin{definition}{Arc Full Membership}
An arc $a=\langle n_{\texttt{src}}$, $n_{\texttt{dst}}$, $\{p\}\rangle$ is a \emph{full member} of the class $a^\alpha=\langle n^\alpha_{\texttt{src}}$, $n^\alpha_{\texttt{dst}}$, $\{p^\alpha\}\rangle$, 
denoted $a \AFM a^\alpha$, iff $a$ is a left and right member of $a^\alpha$, \ie 
	\begin{equation*}
	a \AFM a^\alpha \quad\Longleftrightarrow\quad a \ALM a^\alpha \;\wedge\; a \ARM a^\alpha.
	\end{equation*}
\end{definition}

The \arcKilled arc is not a full member of the \arcClassSuicide class as it is not a left member of this class. The \arcSuicide arc is an example of a full member of the \arcClassSuicide class.

\newtheorem{theorem}{Theorem}
\newtheorem{claim}[theorem]{Claim}
\begin{claim}
An arc $a$ is a full member of a class $a^\alpha$ iff $a$ is an instance of $a^\alpha$. 
\end{claim}

\begin{proof}
\noindent First, if an arc $a$ is a full member of a class $a^\alpha$, then all property constraints are satisfied by the properties ($a \ASM a^\alpha$). Additionally, from the definitions on left and right membership, the values of the properties \texttt{src} and \texttt{dst} statisfy the property constraints with the same name ($n_{\texttt{dst}} \NSM n^\alpha_{\texttt{dst}}$ and $n_{\texttt{src}} \NSM n^\alpha_{\texttt{src}}$). Therefore,
\begin{equation}
a \AFM a^\alpha \Rightarrow a \sqsubset a^\alpha.
\label{eq:fullMembershipImpliesInstance}
\end{equation}

Next, if an arc $a$ is an instance of a class $a^\alpha$, then $a$ is a strict member of $a^\alpha$ (from Eq.~\ref{eq:membershipImpliesInstance}). Additionally, the values of the properties \texttt{src} and \texttt{dst} satisfy the property constraints with the same name ($n_{\texttt{dst}} \NSM n^\alpha_{\texttt{dst}}$ and $n_{\texttt{src}} \NSM n^\alpha_{\texttt{src}}$). Therefore,
\begin{equation*}
a \sqsubset a^\alpha \Rightarrow 
		a \ASM a^\alpha \;\wedge\; 
		n_{\texttt{dst}} \NSM n^\alpha_{\texttt{dst}} \;\wedge\; 
		n_{\texttt{src}} \NSM n^\alpha_{\texttt{src}}\text{, \ie}
\end{equation*}
\begin{equation}
a \sqsubset a^\alpha \Rightarrow a \AFM a^\alpha.
\label{eq:instanceImpliesfullMembership}
\end{equation}

Finally, from Eqs.~\ref{eq:fullMembershipImpliesInstance} and \ref{eq:instanceImpliesfullMembership},
\begin{equation*}
a \AFM a^\alpha \quad\Longleftrightarrow\quad a \sqsubset a^\alpha.
\label{eq:fullMembershipEquivalentInstance}
\end{equation*}
\end{proof}

\begin{definition}{Node Relational Membership}
A node $n$ is a \emph{relational member} of a class $n^\alpha$, denoted $n \NFM n^\alpha$, iff 
%\begin{inparaenum}[(1)] 
\renewcommand{\labelenumi}{(\arabic{enumi})}
\begin{enumerate}
	\item  $n$ is a strict member of $n^\alpha$, 
	\item for each class-based arc starting from $n^\alpha$ (\ie a class whose value of the property constraint \texttt{src} is $\sqsubset n^\alpha$), at least one arc starting from $n$ (\ie an arc whose value of the property \texttt{src} is $n$) is a left member of $n^\alpha$, and 
	\item for each class-based arcs leading to $n^\alpha$ (\ie a class whose value of the property constraint \texttt{dst} is $\sqsubset n^\alpha$), at least one arc leading to $n$ (\ie an arc whose value of the property \texttt{dst} is $n$) is a right member of $n^\alpha$.
\end{enumerate}
%\end{inparaenum}
\end{definition}

\noindent Formally, $n \NFM n^\alpha$ iff
$$
	\begin{cases}
	1) & n \NSM n^\alpha,\\
  2) & \forall a^\alpha=\langle n^\alpha,{n^\alpha}', \{p^\alpha\} \rangle,\; \exists\, a=\langle n,n',\{p\}\rangle \;:\; a \ALM a^\alpha,\\
  3) & \forall a^\alpha=\langle {n^\alpha}',n^\alpha, \{p^\alpha\}\rangle,\; \exists\, a=\langle n',n,\{p\}\rangle \;:\; a \ARM a^\alpha.	
	\end{cases}
$$

\sloppy
\Juliet is a relational member of \ClassMissCapulet. First, \Juliet is a strict member of \ClassMissCapulet. Second, the arc \arcFeelings starting from \Juliet is a left member of \arcClassFeelings, the only arc starting from \ClassMissCapulet. Third, the arc \arcCousin starting from \Juliet is a right member of \arcClassFeelings, the only arc leading to \ClassMissCapulet. Note that the \arcSuicideToo arc is meaningless as regards class relational membership of \Juliet.

\Juliet is not a relational member of \ClassCapulet. Although \Juliet is a strict member of \ClassCapulet, no arc starting from \Juliet is a left member of \arcClassCousin and no arc leading to \Juliet is a right member of \arcClassKilled. 
\fussy

\section{Compliance(s)}

Based on the membership relations defined above, the concept of compliance of an object-based graph with a class-based graph may be defined. Membership relations are ``local'', as they concern a given node. Compliance concerns whole object-based and class-based graphs, and may therefore be considered as ``global''.

An object-based graph is compliant with a given class-based graph if the constraints on the nodes and the arcs among them, \ie constraints defined in the class-based graph, are satisfied by the given object-based graph. As formally presented below, various levels of compliance may be dis\-tin\-guished.

\begin{definition}{Compliance Relation}
Consider an object-based graph $g =  \langle N, A \rangle$ and a class-based graph $g^\alpha = \langle N^\alpha, A^\alpha \rangle$. 
A \emph{compliance relation} \CR is a relation  on $N \times N^\alpha$ such that
\begin{equation}
\label{eq:complianceRelation1}
\forall(n,n^\alpha)\in N\times N^\alpha,\quad n \CR n^\alpha \Rightarrow n \NFM n^\alpha,
\end{equation}
\begin{align}
\label{eq:complianceRelation2}
\forall (a^\alpha=<n^\alpha_{\texttt{src}}, n^\alpha_{\texttt{dst}}, \{p^\alpha\}>)\in A^\alpha,\qquad\qquad\nonumber\\
 \forall (n_{\texttt{src}},n_{\texttt{dst}})\in N\times N : 
 		n_{\texttt{src}} \CR n^\alpha_{\texttt{src}},  
 		n_{\texttt{dst}} \CR n_{\texttt{dst}}^\alpha, \qquad\qquad\nonumber\\
 \exists (a=<n_{\texttt{src}}, n_{\texttt{dst}}, \{p\}>) \in A \;:\; a \AFM a^\alpha,
 \end{align}
\begin{equation}
\label{eq:complianceRelation3}
\forall n^\alpha \in N^\alpha, \exists n \in N \;:\; n \CR n^\alpha. 
\end{equation}
\end{definition}

First, the compliance of a node $n$ with a class $n^\alpha$ implies that the node $n$ is a relational member of the class $n^\alpha$ (cf.~Eq.~\ref{eq:complianceRelation1}). Second, for each class-based arc $a^\alpha$ between two classes $n^\alpha_{\texttt{src}}$ and $n^\alpha_{\texttt{dst}}$, for each objects $n_{\texttt{src}}$ and $n_{\texttt{dst}}$ being compliant with $n^\alpha_{\texttt{src}}$ and $n^\alpha_{\texttt{dst}}$, respectively, there exists an arc $a$ between $n_{\texttt{src}}$ and $n_{\texttt{dst}}$ that is a full member of $a^\alpha$ (cf. Eq.~\ref{eq:complianceRelation2}). Third, for each class $n^\alpha$, at least one object $n$ is compliant with the class (cf. Eq.~\ref{eq:complianceRelation3}).

\begin{definition}{Compliance with a class-based graph}\label{def:complianceRelation}
An ob\-ject-based graph $g = \langle N, A \rangle$ is \emph{compliant} with a class-based graph $g^\alpha = \langle N^\alpha, A^\alpha \rangle$, denoted $g \CR g^\alpha$, iff there exists a compliance relation \CR on $N \times N^\alpha$.
\end{definition}

\begin{figure}[htp]
	\centering
		%%%%%%%%%%%%%%%%%%%%%%%%%%%
%%
%% Object-based graph
\begin{tikzpicture}[->,>=stealth']

	\node[object] (O1)
	{   \begin{tabular}{l} 	% content
		  \textbf{Romeo}\\
		  \hspace{0.2cm}{$\langle$name, Romeo$\rangle$}\\
		  \hspace{0.2cm}{$\langle$house, Montague$\rangle$}\\
      \hspace{0.2cm}{$\langle$sex, male$\rangle$}\\
		  \end{tabular}
 	};

	\node[object,
	above right of=O1,
  yshift=+2cm,
  xshift=+2cm,
  anchor=center] (O2)
	{
 \begin{tabular}{l} 	% content
  \textbf{Mercutio}\\
  \hspace{0.2cm}{$\langle$name, Mercutio$\rangle$}\\
  \hspace{0.2cm}{$\langle$house, Verona$\rangle$}\\
  \hspace{0.2cm}{$\langle$sex, male$\rangle$}\\
 \end{tabular}
 };

	\node[object,
	below right of=O2,
  yshift=-2.2cm,
  xshift=2cm,
  anchor=center] (O3)
	{
 \begin{tabular}{l} 	% content
  \textbf{Juliet}\\
  \hspace{0.2cm}{$\langle$name, Juliet$\rangle$}\\
  \hspace{0.2cm}{$\langle$house, Capulet$\rangle$}\\
  \hspace{0.2cm}{$\langle$sex, female$\rangle$}\\
  \hspace{0.2cm}{$\langle$age, 13$\rangle$}
 \end{tabular}
 };

 \path 
 		(O1) 	edge[bend left=20]  node[arc,above left=-5mm]{ \begin{tabular}{l}
																														\textit{friend}\\
																														$\langle$friend, xxx$\rangle$\\
																														\end{tabular}}(O2)
		(O1) 	edge[bend left=12]  node[arc,above=1mm]{\begin{tabular}{l}
																														\textit{feelings}\\
																														$\langle$feels,Love$\rangle$\\
																														\end{tabular}}
		(O3) 	
		
		(O3) edge[bend left=12] node[]{} (O1)

		(O1)  edge[loop below,min distance=11mm]    node[arc,below](suic1){\begin{tabular}{l}
																														\textit{commit suicide}\\
																														$\langle$killing, poison$\rangle$
																														\end{tabular}} 
		(O1)
 																														
		(O3)  edge[loop below,min distance=5mm]    node[arc,below]{\begin{tabular}{l}
																														\textit{commit suicide too}\\
																														$\langle$killing, dagger$\rangle$
																														\end{tabular}} 
		(O3)
;
\end{tikzpicture}
	\caption{Example of an object-based graph partially compliant with the class-based graph presented in Fig.~\ref{fig:classGraph}.}
	\label{fig:objectGraphCompliant}
\end{figure}
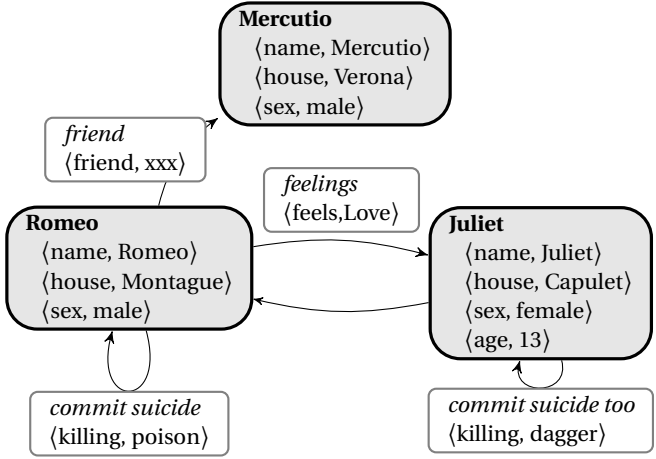

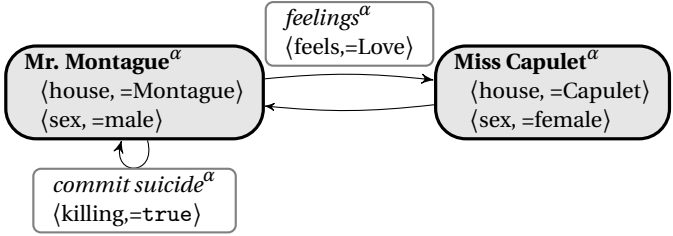
\begin{figure}[htp]
	\centering
		%%%%%%%%%%%%%%%%%%%%%%%%%%%
%%
%% Object-based graph
\begin{tikzpicture}[->,>=stealth']

	\node[object] (C1)
	{   \begin{tabular}{l} 	% content
		  \textbf{Mr. Montague$^\alpha$}\\
		  \hspace{0.2cm}{$\langle$house, =Montague$\rangle$}\\
  		\hspace{0.2cm}{$\langle$sex, =male$\rangle$}\\
		  \end{tabular}
 	};

	\node[object,
	right of=C1,
  xshift=+4.5cm,
  anchor=center] (C3)
	{
 \begin{tabular}{l} 	% content
  \textbf{Miss Capulet$^\alpha$}\\
  \hspace{0.2cm}{$\langle$house, =Capulet$\rangle$}\\
  \hspace{0.2cm}{$\langle$sex, =female$\rangle$}\\
 \end{tabular}
 };

 \path 
		(C1) 	edge[bend left=6]  node[arc,above=1mm]{\begin{tabular}{l}
																														\textit{feelings$^\alpha$}\\
																														$\langle$feels,=Love$\rangle$\\
																														\end{tabular}}
		(C3) 	
		
		(C3) edge[bend left=6] node[]{} (C1)

		(C1)  edge[loop below,min distance=5mm]    node[arc,below]{\begin{tabular}{l}
																														\textit{commit suicide$^\alpha$}\\
																											$\langle$killing,=\texttt{true}$\rangle$\\
																														\end{tabular}} 
		(C1)
 																														;

\end{tikzpicture}
	\caption{Example of a class-based graph with which the object-based graph presented in Fig.~\ref{fig:objectGraphCompliant} is compliant.}
	\label{fig:classGraphCompliant}
\end{figure}

%\sloppy
To illustrate compliance, consider the object-based graph presented in Fig.~\ref{fig:objectGraphCompliant} and the class-based graph presented in Fig.~\ref{fig:classGraphCompliant}. The relation \CR, such that \Romeo \CR \ClassMrMontague and \Juliet \CR \ClassMissCapulet, is a compliance relation.
First, \Romeo (resp. \Juliet) is a relational member of \ClassMrMontague (resp. \ClassMissCapulet). Second, for all class-based arcs (\arcClassFeelings and \arcClassSuicide), full member arcs (\arcFeelings and \arcSuicide) exist. Finally, there is no class without a compliant object.

Note that the \texttt{Mercutio} node is meaningless as regards compliance of the two considered graphs. Therefore, additional nodes may be added to the object-based graph without changing its compliance with the class-based graph. A similar remark concerns arcs, such as \arcSuicideToo.
%\fussy

\begin{definition}{Partial compliance relation}
Consider a class-based graph $g^\alpha=\langle N^\alpha, A^\alpha\rangle$ and an object-based graph $g = \langle N, A\rangle$. 
A \emph{partial compliance relation} \PCR on $N \times N^\alpha$ is a relation that satisfies only the conditions of Eqs.~\ref{eq:complianceRelation1} and \ref{eq:complianceRelation2}, the condition of Eq.~\ref{eq:complianceRelation3} being relaxed.
\end{definition}

\begin{definition}{Partial compliance with a class-based graph}\label{def:partialComplianceRelation}
An object-based graph $g=\langle N, A\rangle$ is \emph{partially compliant} with a class-based graph $g^\alpha=\langle N^\alpha, A^\alpha \rangle$, denoted $g \PCR g^\alpha$, iff there exists a partial compliance relation \PCR on $N \times N^\alpha$.
\end{definition}

To illustrate partial compliance, consider the object-based graph presented in Fig.~\ref{fig:objectGraphCompliant} and the class-based graph presented in Fig.~\ref{fig:classGraph}. The relation \PCR, such that \Romeo \PCR \ClassMrMontague and \Juliet \PCR \ClassMissCapulet, is a partial compliance relation.

First, \Romeo (resp. \Juliet) is a relational member of \ClassMrMontague (resp. \ClassMissCapulet). Second, for all class-based arcs (\arcClassFeelings and \arcClassSuicide), full member arcs (\arcFeelings and \arcSuicide) exist. However, there is no node compliant with the class \ClassCapulet.

The difference between partial compliance and normal compliance is the relaxation of Eq.~~\ref{eq:complianceRelation3}. Therefore, in the normal compliance case, for each class, there should be at least one object being a member of this class, while in a partial compliance case, some class may not have any relational member object, \eg the class \ClassCapulet from Fig.~\ref{fig:classGraph} with regard to the object-based graph presented in Fig.~\ref{fig:objectGraphCompliant}.

\begin{definition}{Full compliance relation}
A \emph{full compliance relation} \FCR on $N \times N^\alpha$ is a compliance relation such that
\begin{equation}
\label{eq:fullComplianceRelation1}
\forall n \in N , \exists n^\alpha \in N^\alpha \;:\; n \FCR n^\alpha. 
\end{equation}
\vspace{-24pt}
\end{definition}

\begin{definition}{Full Compliance with a class-based graph}\label{def:fullComplianceRelation}
An object-based graph $g = \langle N, A \rangle$ is \emph{fully compliant} with a class-based graph $g^\alpha= \langle N^\alpha, A^\alpha \rangle$, denoted $g \FCR g^\alpha$, iff there~exists a full compliance relation \FCR on $N \times N^\alpha$.
\end{definition}

To illustrate full compliance, consider the object-based graph presented in Fig.~\ref{fig:objectGraph} and the class-based graph presented in Fig.~\ref{fig:classGraph}. The relation \FCR, such that  \Romeo \FCR \ClassMrMontague, \Tybalt \FCR \ClassCapulet, and  \Juliet \FCR \ClassMissCapulet, is a full compliance relation.

First, \Romeo (resp. \Tybalt, \Juliet) is a relational member of \ClassMrMontague (resp. \ClassCapulet, \ClassMissCapulet). Second, for all class-based arcs, full member arcs exist. 
Third, there is no class without a compliant node. The relation \FCR is therefore a compliance relation. Additionally, all the nodes are compliant with a class. As a conclusion, the relation \FCR is a full compliance relation.

Although each object-based graph fully compliant with a class-based graph is by definition compliant too, the opposite is not true. A compliant object-based graph may contain objects that are not members (neither strict not relational) of any classes of the class-based graph. Such an object-based graph is not fully compliant with a class-based graph as it does not satisfy Eq.~\ref{eq:fullComplianceRelation1}. As an example, the object-based graph presented in Fig.~\ref{fig:objectGraphCompliant} is compliant but not fully compliant with the class-based graph presented in Fig.~\ref{fig:classGraphCompliant}, as there is no class \Mercutio is a member of. 
\section{Related Works}
\label{sec:relatedWorks}

Among the related works, object-oriented languages, data\-base schemata, and ontologies have already mentioned in the introduction. A main drawback of these approaches with regard to the addressed problem is their limited support for arcs, with a limited set of predefined relations among nodes, such as \texttt{has-a} and \texttt{is-a} relations in the object-oriented  paradigm and joins in relational databases.

The Entity-Relationship (ER) model proposed by Chen~\cite{journal:atds:Chen:1976} is a data model in many aspects similar to one proposed in this paper. In the ER model, an entity is defined in similar to the concept of object: ``the information about an entity or a relationship [\ldots] is expressed by a set of attribute-value pairs''. Note that note only the entities, similarly to object-based nodes, are defined by attributes (or properties). Relationships, similarly to object-based arcs, are defined by their attributes too.

The concept of class is expressed in the ER model via ``sets'': entity sets are similar to class-based nodes. Relationship sets are similar to class-based arcs. Attribute sets are similar to property constraints. The ER model is more flexible with regard to relationships. The ER model allows for the definition of $n$-ary relationships, \ie relationships connecting more than two entities. The model proposed in this paper is limited to arcs among two nodes.

The ER model is more restrictive than our model with regard to the relation between entities and entity sets. In the ER model, all the attributes of an entity that to match the attribute sets of the associated entity set. In our model, an object may contain a property that does not satisfy any property constraint of the class the given object is an instance of.

Finally, in the ER model, the list of predicates to restrict the entities and the relationships is limited to constraints on allowable values for a value set, constraints on permitted values for a certain attribute, constraint between sets of existing values, and constraints between particular values. In the proposed example, any predicate may be used to constraint the values of properties, either for nodes, or for arcs.

In the area of knowledge representation, most proposed models are based on graph-based ontologies, such as RDF~\cite{w3c:RDF-PRIMER} or OWL~\cite{w3c:owl2primer}. In RDF, ``the things being described have properties which have values'', similarly to object-based nodes. The RDF Schema (RDFS) recommendation~\cite{w3c:rdfs}  defines a limited set of ``classes and properties that may be used to describe classes, properties and other resources''. RDF, combined with RDFS, does not provide the modelling power provided by our model with regard to the typing of relationships among objects, as RDF allows only for named relationships, without the possibility to attach a set of attributes. 

RDF and RDFS supports generalization and specialization of classes. The only relation considered between classes and objects is the \texttt{rdf:type} property that  is used to state that a resource is an instance of a class. No definition of compliance or similar global concepts is proposed in RDF and RDFS.

Similarly, OWL, as an extension of RDF and RDFS, supports the modelling of classes and objects with attributes. The OWL Full variant (the most complex OWL variant) allows relations to be objects, \ie relations among objects may be described with objects described by a set of properties. However, similarly to RDF and RDFS, OWL defines only the \texttt{type} property to connect objects and classes, and neither compliance nor similar global concepts is proposed.

Finally, the proposed model may be compared with a group of languages and protocols aiming at supporting social networks. The ontology Friend-of-a-Friend (shorten as FOAF) aims at describing persons and objects, as well as their relations. In FOAF, a list of classes and properties associated mainly to individuals, documents, multimedia data, and online activity are standardized. FOAF is based on RDF and OWL, and therefore has inherited some limitations from these standards: relations may not be described with attributes or properties, and the \texttt{type} property is the only property linking object and classes.

Another related approach is the Activity Streams protocol~\cite{www:activityStream} aiming at providing an aggregate view of the activities performed by individuals across the social websites they are interacting with.  In Activity Streams, an activity consists of an actor, an action performed by the actor (a \emph{verb}, a thing (an \emph{object)} that is the actor is performing his/her action against, and eventually a target involved. The Activity Streams protocol defines a list of standard verbs and object types.  In Activity Streams, verbs are not described with attributes and the list of standard object types is specified in the Activity Base Schema~\cite{www:activityStreamBaseSchema}.

Another recently proposed approach supporting social networks is the Open Graph protocol~\cite{www:opengraph}. In Open Graph, a model for user activities based on the concept of actions and objects is proposed. In a similar manner to RDF, activities in Open Graph are triplets <individual, action, object>. However, when RDF models actions without properties, the Open Graph actions may have attributes, as well as individuals and objects. However, Open Graph focuses on the \emph{creation of the social network}, while our approach is prescriptive, aiming at defining constraints (as class-based arcs and property predicates) to improve the \emph{identification} of a network compliant with a network schema.
\section{Conclusions}

In this paper, object-based and class-based graphs are formally defined. The concepts of membership and compliance are proposed as special relations between nodes and classes, either locally or globally.

The three types of compliance---partial, normal, and full compliance---correspond to three different situations. Partial compliance may be useful in the situation when nodes compliant with classes are progressively identified, when some classes may not have compliant nodes. An application may be the support for emergency teams going to a emergency site: some members of the team may already be on the emergency site, while others are on their way. A class-based graph may be used to check that the partial team is compliant. Normal compliance may be useful in the situation when all the classes have to be associated with a compliant node, even if some additional nodes may exist in the object-based graph. An example may be the verification that all the services of a service-oriented application are compliant with a class-based graph defining the type of services needed and their relations. Finally, full compliance may be useful in the situation when each class has to be associated with a compliant node, and each node has to be compliant with a class. An example may be the specification of controlled chemical reaction as class-based graphs. In such a situation, no additional chemical substance may be added to the graph of chemical substance directly participating in the chemical reaction.

Among future works, the development of algorithms to check various types of compliances of object-based graphs with a given class-based graph is still an open issue. A major issue in the development of such algorithms is scalability. These algorithms should be adapted to large modern networks, such as Facebook that consists of hundreds of millions of users, and therefore they should be efficient and scalable.

Another area of improvement is the support for inheritance in class-based graphs.	In the presented example, the class \ClassMissCapulet is a specialized class of the \ClassCapulet. Providing support for class inheritance would lead to a more expressive and concise representation of class-based graphs.

\section*{References}
\bibliography{biblio}
\bibliographystyle{model1-num-names}

\end{document}